\documentclass{article}
\usepackage{amsfonts}
\usepackage{amsmath}
\usepackage{graphicx}
\usepackage{mathdots}
\usepackage{tikz}
\usepackage{amsthm}
\usepackage{url,fullpage}

\newtheorem{theorem}{Theorem}[section]
\newtheorem*{theorem*}{Theorem}
\newtheorem{proposition}{Proposition}[section]
\newtheorem*{proposition*}{Proposition}
\theoremstyle{definition}
\newtheorem{definition}{Definition}[section]
\newtheorem*{definition*}{Definition}
\newtheorem{problem}{Problem}[section]
\newtheorem*{problem*}{Problem}
\newtheorem{example}{Example}[section]
\newtheorem*{example*}{Example}
\newtheorem{assumption}{Assumption}[section]
\newtheorem*{assumption*}{Assumption}

\newtheorem*{remark*}{Remark}
\usetikzlibrary{positioning}

\newcommand{\mvec}[1]{\underline{{#1}}}

\title{An Improvement of a Key Exchange Protocol Relying on Polynomial Maps}
\author{Keita Suzuki${}^{1}$ \and Koji Nuida${}^{23}$}
\date{%
${}^1$ Graduate School of Information Science and Technology, The University of Tokyo, Japan \\
\url{keita-suzuki391@g.ecc.u-tokyo.ac.jp} \\
${}^2$ Institute of Mathematics for Industry (IMI), Kyushu University, Japan \\
\url{nuida@imi.kyushu-u.ac.jp} \\ 
${}^3$ National Institute of Advanced Industrial Science and Technology (AIST), Japan%
}

\begin{document}
\maketitle

\begin{abstract}
Akiyama et al.\ (Int.~J.~Math.~Indust., 2019) proposed a post-quantum key exchange protocol that is based on the hardness of solving a system of multivariate non-linear polynomial equations but has a design strategy different from ordinary multivariate cryptography.
Their protocol has two versions, an original one and a modified one, where the modified one has a trade-off that its security is strengthened while it has non-zero error probability in establishing a common key.
In fact, the evaluation in their paper suggests that the probability of failing to establish a common key by the modified protocol with the proposed parameter set is impractically high.
In this paper, we improve the success probability of Akiyama et al.'s modified key exchange protocol significantly while keeping the security, by restricting each component of the correct common key from the whole of the coefficient field to its small subset.
We give theoretical and experimental evaluations showing that our proposed parameter set for our protocol is expected to achieve both failure probability $2^{-120}$ and $128$-bit security level.
\end{abstract}

\section{Introduction}

Cryptographic schemes are mainly classified into public key cryptography and symmetric key cryptography.
In symmetric key encryption schemes, both the sender and the receiver of a message use a common key, therefore they have to share the common key in advance.
Public key cryptographic schemes used for this purpose are called key exchange protocols.
The famous key exchange protocols include RSA-based one \cite{RSA1} and Diffie--Hellman key exchange \cite{Diffie}.

For the widely used public key cryptosystems at the present such as RSA cryptosystem \cite{RSA2} and elliptic curve cryptosystems \cite{elliptic,elliptic2}, it has been known that Shor's quantum algorithm \cite{Shor} can break these schemes in polynomial time.
This means that those cryptosystems will be vulnerable once a large-scale quantum computer is developed in future.
As a countermeasure, post-quantum cryptography (PQC) has been studied intensively, which is the class of (public key) cryptographic schemes secure even against quantum algorithms.
NIST has started the standardization process for PQC since 2017, and a lot of candidates have been submitted.
Among them, one of the main candidates for PQC is multivariate cryptography, which is based on the hardness of solving a system of multivariate non-linear equations.

In 2019, Akiyama et al.~\cite{Akiyama} proposed a key exchange protocol which is based on the hardness of solving a system of multivariate non-linear equations but yet has a design strategy different from ordinary multivariate cryptography.
Namely, ordinary multivariate cryptography intends to conceal the internal structure by using compositions of maps, while Akiyama et al.'s scheme uses both composition and addition of maps to conceal the central map $\mvec{\psi}$.
Their construction improved a previous protocol \cite{Yosh} in reducing the ciphertext size.
Now we note that their first construction of the proposed protocol has a restriction that the central map $\mvec{\psi}$ must be injective, which decreases the number of candidates for $\mvec{\psi}$ and hence weakens the security.
As a countermeasure, they also proposed a modified construction of the protocol that enables to use non-injective central map $\mvec{\psi}$ and hence strengthens the security.
However, this modification yields a trade-off that the modified protocol has non-zero error probability in establishing a common key (in contrast to the original version where such an error does not occur), and the error probability of the modified protocol is expected as high.
In fact, for the experiments of the modified protocol performed in the original paper \cite{Akiyama}, the highest average success rate with their proposed parameters was only $89.9\%$.
In this paper, we focus on the modified version of the protocol in \cite{Akiyama} and aim at resolving the high error probability to make the scheme practically useful.

\subsection{Our Contributions}

In this paper, we give an improvement of the key exchange protocol by Akiyama et al.~\cite{Akiyama}, more precisely its modified version described in that paper mentioned above, in reducing the failure probability while keeping the security level.
Roughly speaking, in the modified version of Akiyama et al.'s protocol, one of the two parties solves a certain system of polynomial equations defined over a prime field $\mathbb{F}_q$ in order to determine the common key.
The main reason of the larger failure probability is that the system of equations frequently has two or more solutions in $\mathbb{F}_q^n$ and therefore the correct common key is not uniquely determined.
Our main idea is to restrict the range of the correct common key (i.e., the correct solution of the system of equations) into $\mathbb{Z}_p^n \subseteq \mathbb{F}_q^n$ for smaller $p < q$ instead of the whole of $\mathbb{F}_q^n$.
(This idea is inspired by the construction of \cite{Akiyama2}.)
Even if the system of equations has multiple solutions in $\mathbb{F}_q^n$, the solution in $\mathbb{Z}_p^n$ will be unique with high probability, which enables the party to successfully determine the correct common key.
Our proposed protocol with sizes of parameters $q$ and $n$ similar to those of \cite{Akiyama} indeed reduces the failure probability significantly.
Moreover, we propose a parameter set for our proposed protocol that might achieve both failure probability $2^{-120}$ and $128$-bit security level.
We theoretically estimate an upper bound for the failure probability in order to confirm the former property, and discuss the security level against some typical kinds of attacks (Gr\"{o}bner basis computation, linear algebraic attack, etc.) in order to confirm the latter property.

\subsection{Organization of This Paper}

Section \ref{sec:preliminaries} summarizes some notations and basic properties such as properties of Gr\"{o}bner basis.
In Section \ref{sec:previous_scheme}, we recall the construction of the protocol in \cite{Akiyama} and summarize its advantages and disadvantages.
In Section \ref{sec:proposed_protocol}, we describe our proposed protocol and give a theoretical estimate of an upper bound for failure probability of the protocol.
Section \ref{sec:experimental_results} summarizes the results of our computer experiments about our proposed protocol.
Section \ref{sec:comparison} summarizes the results of comparison between our proposed protocol and the protocol in \cite{Akiyama}.
Finally, in Section \ref{sec:security_evaluation}, we evaluate the security of our proposed protocol.

\section{Preliminaries}
\label{sec:preliminaries}

In this section, we recall algebraic definitions needed to understand the algorithms, and introduce PME problem, which is the base of the security of the existing method.

\subsection{Notation}

Let $\mathbb{F}_q$ denote the finite field of $q$ elements and $\mvec{x} := (x_1, \dots, x_n)$ for $n$ variables $x_1,\dots,x_n$.
Then, $f(\mvec{x})$ expresses a polynomial map of $n$ variables.
We refer to $\mathbb{F}_q[\mvec{x}]$ as the polynomial ring of $n$ variables on coefficient ring $\mathbb{F}_q$, and express $(\psi_1(\mvec{x}),\dots,\psi_m(\mvec{x}))$ as $\mvec{\psi}(\mvec{x})$ for polynomial maps of $n$ variables $\psi_1(\mvec{x}),\dots,\psi_m(\mvec{x})$.
Then we define the degree of $\mvec{\psi}$ by $\deg \mvec{\psi} := \max\{\deg \psi_1,\dots,\deg \psi_n\}$.
In particular, polynomial maps of degree one are called affine maps.

For polynomial maps $\mvec{\psi}=(\psi_1,\dots,\psi_m)$ of $n$ variables and $\mvec{\phi}=(\phi_1,\dots,\phi_n)$, we define the composed map $\mvec{\psi}\circ \mvec{\phi}$ by $\mvec{\psi}\circ \mvec{\phi} = (\psi_1(\phi_1,\dots,\phi_n),\dots,\psi_m(\phi_1,\dots,\phi_n))$.
In this paper, we use the following classes of polynomials:
\[
\begin{split}
\Lambda_{n,d} &:= \{f \in \mathbb{F}_q[x_1, \dots, x_n] \mid \deg f = d\} \enspace,\\
\Lambda_{n,d}^m &:= \{\mvec{f} \in \mathbb{F}_q[x_1, \dots, x_n]^m \mid \deg \mvec{f} = d\} \enspace,\\
(\Lambda_{n,d}^n)^* &:= \{\mvec{\psi} \in \Lambda_{n,d}^n \mid \mvec{\psi} \mbox{ is injective}\} \enspace.
\end{split}
\]
Moreover, $\xleftarrow{r}S$ represents that we choose randomly from the set $S$.

\subsection{PME Problem}
\label{subsec:PME_problem}

PME problem was introduced by Akiyama et al.~\cite{Akiyama}, which is the base of the security of their previous method.
The definition is as follows.

\begin{definition}
[PME problem]
For $(f(\mvec{x}), c_1(\mvec{x}), \dots, c_n(\mvec{x})) \in \mathbb{F}_q[\mvec{x}]^{n+1}$ and $(u_1, \dots, u_n) \in \mathbb{F}_q^n$, 
PME (Polynomial Map Equation) problem is a problem of finding a solution to the system of multivariate polynomial equations
\begin{eqnarray*}
  \begin{cases}
    f(x_1, \cdots, x_n) = 0 \\
    c_1(x_1, \cdots, x_n) = u_1\\
    \qquad\vdots\\
    c_n(x_1, \cdots, x_n) =u_n \enspace.\\
  \end{cases}
\end{eqnarray*}
\end{definition}

\subsection{Gr\"{o}bner Basis}

Here we summarize some basics for Gr\"{o}bner basis.
See e.g., \cite{text} for the details.

First, we define the monomial term set $\mathcal{M}_n$ of polynomial ring $K[\mvec{x}]$ on coefficient field $K$:
\begin{eqnarray*}
    \mathcal{M}_n:=\{x_1^{\alpha_1}\cdots x_n^{\alpha_n}\mid(\alpha_1, \cdots, \alpha_n) \in (\mathbb{Z}_{\geq 0})^n\} \enspace.
\end{eqnarray*}

\begin{definition}
[monomial order]
For a strict partial order $\succ$ on $\mathcal{M}_n$, $\succ$ is called monomial order when it satisfies the following three conditions.
\begin{enumerate}
\item
$\succ$ is a total order.
\item
Any subset $S \neq \emptyset$ of $\mathcal{M}_n$ has the minimum element about $\succ$.
\item
If $s,t \in \mathcal{M}_n$ and $s \succ t$, then for any element $u$ of $\mathcal{M}_n$, we have $su \succ tu$.
\end{enumerate}
\end{definition}

\begin{example}
[lexicographical order]
We define the lexicographical order $\succ_{lex}$ on $\mathcal{M}_n$ by
\[
x_1^{\alpha_1} \cdots x_n^{\alpha_n} \succ_{lex} x_1^{\beta_1} \cdots x_n^{\beta_n} \Leftrightarrow 1 \leq \exists i \leq n \mbox{ s.t.\ } \alpha_j = \beta_j \mbox{ ($1 \leq j \leq i-1$) and } \alpha_i > \beta_i \enspace.
\]
Then it is known that $\succ_{lex}$ is a monomial order.
\end{example}

\begin{definition}
[leading term]
Let $f \in K[\mvec{x}]$.
The leading term of $f$ is the maximum monomial (with coefficient) in $f$ with respect to a given monomial order $\succ$, and it is expressed as $LT_{\succ}(f)$.
\end{definition}

Then the following holds about division of multivariate polynomials.

\begin{proposition}
Let $\mathcal{G}=\{g_1,\dots,g_s\}$ be a finite subset of $K[\mvec{x}]\setminus \{0\}$.
For each $f \in K[\mvec{x}]$, there exist $h_1,\dots,h_s,r \in K[\mvec{x}]$ satisfying the following:
\begin{itemize}
\item
$f = h_1 g_1 + \cdots h_s g_s + r$.
\item
Each monomial appearing in $r$ is not divisible by $LT_{\succ}(g_i)$ for any $1 \leq i \leq s$.
\end{itemize}
\end{proposition}

In the above situation, we write the term $r$ as $\overline{f}^{\succ, \mathcal{G}}$ and call it remainder of $f$ about $\succ$ divided by $\mathcal{G}$.
Now Gr\"{o}bner basis is defined as follows.

\begin{definition}
[Gr\"{o}bner basis]
A finite subset $\mathcal{G}=\{g_1,\dots,g_s\}$ of $K[\mvec{x}] \setminus \{0\}$ is a Gr\"{o}bner basis of an ideal $\mathcal{I} \subseteq K[\mvec{x}]$ about $\succ$ when it satisfies the following condition: For any $f \in \mathcal{I} \setminus \{0\}$, $LT_{\succ}(g_i)$ divides $LT_{\succ}(f)$ for some $1 \leq i \leq s$.
\end{definition}

We note that this definition is equivalent to saying that for any $f \in \mathcal{I} \setminus \{0\}$, the remainder always satisfies that $\overline{f}^{\succ, \mathcal{G}} = 0$.

Hereinafter, we consider solving a system of multivariate non-linear polynomial equations using Gr\"{o}bner basis.
First, we formulate the problem as follows.

\begin{problem}
Given $s$ polynomials $f_1,\dots,f_s \in K[\mvec{x}]$ in $n$ variables $\mvec{x}$, find $(a_1,\dots,a_n) \in K^n$ such that $f_1(a_1,\dots,a_n)= \cdots =f_s(a_1,\dots,a_n)=0$.
\end{problem}

We reformulate this problem in a way that we can apply Gr\"{o}bner basis to the problem.

\begin{problem}
\label{prob:zero_set}
Let $\mathcal{I} \subseteq K[\mvec{x}]$ be the ideal generated by $f_1,\dots,f_s \in K[\mvec{x}]$.
Then find the zero set $V(\mathcal{I}) := \{(a_1,\cdots,a_n)\in K^n \colon f(a_1,\dots,a_n)=0 \mbox{ for any } f \in \mathcal{I}\}$.
\end{problem}

For Gr\"{o}bner basis about lexicographical order, the following property is useful.

\begin{theorem}
For lexicographical order $\succ_{lex}$ with $x_1 \succ x_2 \succ \cdots \succ x_n$, let $\mathcal{G}$ be a Gr\"{o}bner basis of an ideal $\mathcal{I} \subseteq K[\mvec{x}]$ about $\succ_{lex}$.
Then, for any $1 \leq \ell \leq n$, $\mathcal{G} \cap K[x_{\ell}, \dots,x_n]$ is a Gr\"{o}bner basis of the ideal $\mathcal{I} \cap K[x_{\ell}, \dots,x_n]$ of $K[x_{\ell}, \dots,x_n]$.
\end{theorem}

This theorem enables us to reduce Problem \ref{prob:zero_set} to solving problems in less variables.
Moreover, when the ideal $\mathcal{I}$ is a zero-dimensional ideal (in the sense explained below), we can reduce the problem to a further easier one.

\begin{definition}
[zero-dimensional ideal]
An ideal $\mathcal{I} \subseteq K[\mvec{x}]$ is a zero-dimensional ideal when the quotient space $K[\mvec{x}] / \mathcal{I}$ is a finite-dimensional linear space over $K$.
\end{definition} 

A Gr\"{o}bner basis $\mathcal{G}= \{g_1,\cdots,g_s\}$ of a zero-dimensional ideal $\mathcal{I}$ about lexicographical order satisfies (with a certain ordering for the elements of $\mathcal{G}$) that for each $1 \leq i \leq n$, $g_i$ is a polynomial in $x_1,\dots,x_i$.
This fact makes it possible to solve the original system of equations by solving univariate non-linear equations finitely many times.

Buchberger's algorithm, $F_4$ algorithm \cite{F4}, and $F_5$ algorithm \cite{F5} are frequently used to calculate Gr\"{o}bner basis.
When the ideal in the problem is a zero-dimensional ideal, we can estimate the computational complexity of $F_5$, which is the best algorithm among them, by using the notion of degree of regularity explained below.

\begin{definition}
[degree of regularity]
We define the degree of regularity of a zero-dimensional ideal $\mathcal{I}=\langle f_1,\dots,f_s\rangle$ by
\begin{eqnarray*}
d_{reg} := \min\left\{ d\geq 0\mid \dim \{ f \in I \colon f \mbox{ is homogeneous of degree } d \} = \binom{n+d-1}{d} \right\} \enspace.
\end{eqnarray*}
\end{definition}

\begin{definition}
[$d$-regular]
For an overdetermined system of polynomial equations $f_1 = \cdots = f_s = 0$ ($s \geq n$) whose polynomials generate a zero-dimensional ideal, this equation system is $d$-regular when the following holds for any $1 \leq i \leq s$ and $g \in K[\mvec{x}]$:
\begin{center}
If $\deg(g) < d-\deg(f_i)$ and $gf_i \in \langle f_1,\dots,f_{i-1}\rangle$, then $g \in \langle f_1,\dots,f_{i-1}\rangle$.
\end{center}
\end{definition}

\begin{definition}
[semi-regular]
A system of polynomial equations is semi-regular when it is $d_{reg}$-regular.
\end{definition}

The following result is shown in \cite{Grobner}.

\begin{theorem}
\label{thm:complexity_of_F5}
For a semi-regular system of polynomial equations, the complexity of $F_5$ algorithm is estimated as
\begin{eqnarray*}
O\left({\binom{n+d_{reg}}{n}}^{\omega}\right)
\end{eqnarray*}
where $\omega<2.39$ denotes the linear algebra constant.
\end{theorem}

\section{The Previous Protocol}
\label{sec:previous_scheme}

In this section, we summarize the protocol proposed by Akiyama et al.~\cite{Akiyama}.

\subsection{The Original Protocol}

Here we describe the original version of the protocol given in Section 4 of \cite{Akiyama}.
In the protocol, Alice and Bob are going to agree on a common key using a public channel.
We use the following parameters:
\begin{align*}
q \colon& \mbox{prime number which is the number of elements of the coefficient field}\\
n \colon& \mbox{number of variables}\\
m \colon& \mbox{degree of polynomials generated by Bob}\\
d \colon& \mbox{degree of polynomal generated by Alice}
\end{align*}
The protocol is as follows (see also Figure \ref{fig:previous_protocol}).
\begin{enumerate}
\item \label{previous_algorithm__Alice}
Alice sends a multivariate equation $f(\mvec{x})=0$ to Bob and keeps its solution $\mvec{\sigma} \in \mathbb{F}_q^n$ secret.
The detail is as follows:
\begin{enumerate}
\item
Generate a polynomial $f(\mvec{x}) \in \mathbb{F}_q[\mvec{x}]$ of degree $d$ uniformly at random.
\item
Generate a solution $\mvec{\sigma} \in \mathbb{F}_q^n$ of $f(\mvec{x})=0$ as follows.
First, generate $\sigma_1, \dots, \sigma_{n-1} \in \mathbb{F}_q$ uniformly at random.
Then, solve a univariate equation $f(\sigma_1, \dots, \sigma_{n-1}, x_n)=0$ in $x_n$ and keep a solution $\sigma_n$; if it cannot be solved, then modify the constant term and restart generating $\sigma_1, \dots, \sigma_{n-1}$.
\item
Keep the solution $\mvec{\sigma} = (\sigma_1, \dots, \sigma_{n-1}, \sigma_n) \in \mathbb{F}_q^n$.
\item
Send $f(\mvec{x})$ to Bob.
\end{enumerate}
\item
Bob sends multivariate polynomials $\mvec{g}(\mvec{x})$ and $\mvec{c}(\mvec{x})$ to Alice.
The detail is as follows:
\begin{enumerate}
\item
Generate a bijective affine map $\mvec{g}(\mvec{x}) = (g_1(\mvec{x}), \dots, g_n(\mvec{x})) \in \mathbb{F}_q[\mvec{x}]^n$ uniformly at random.
\item
Generate an injective polynomial map $\mvec{\psi}(\mvec{x}) = (\psi_1(\mvec{x}), \dots, \psi_n(\mvec{x})) \in \mathbb{F}_q[\mvec{x}]^n$ of degree $\deg \psi_j = m$ ($1 \leq j \leq n$) randomly.
\item
Compute $\mvec{\psi}(\mvec{g}(\mvec{x}))$.
\item
Generate a polynomial map $\mvec{r}(\mvec{x}) = (r_1(\mvec{x}), \dots, r_n(\mvec{x})) \in \mathbb{F}_q[\mvec{x}]^n$ of degree $\deg r_j = m-d$ ($1 \leq j \leq n$) uniformly at random.
\item
Compute a polynomial map $\mvec{c}(\mvec{x}) = \mvec{\psi}(\mvec{g}(\mvec{x})) + f(\mvec{x})\mvec{r}(\mvec{x})$.
\item
Send $\mvec{g}(\mvec{x})$ and $\mvec{c}(\mvec{x})$ to Alice.
\end{enumerate}
\item
Alice computes a common key $\mvec{s} \in \mathbb{F}_q^n$, and sends $\mvec{u} \in \mathbb{F}_q^n$ to Bob.
The detail is as follows:
\begin{enumerate}
\item
Compute $\mvec{g}(\mvec{\sigma}) = \mvec{s}$, and keeps $\mvec{s}$ as a common key.
\item
Compute $\mvec{c}(\mvec{\sigma}) = \mvec{u}$, and send $\mvec{u}$ to Bob.
\end{enumerate}
\item \label{previous_algorithm__Bob}
Bob computes a common key $\mvec{s}$ as follows.
Since $f(\mvec{\sigma})= 0$ implies $\mvec{c}(\mvec{\sigma}) = \mvec{\psi}(\mvec{g}(\mvec{\sigma})) = \mvec{u}$, Bob can compute the common key $\mvec{s}$ by applying $\mvec{\psi}^{-1}$ to $\mvec{u}$:
\begin{eqnarray*}
\mvec{\psi}^{-1}(\mvec{u}) = \mvec{g}(\mvec{\sigma}) =\mvec{s} \enspace.
\end{eqnarray*}
\end{enumerate}
\begin{figure}[t!]
\centering
\begin{tabular}{|lll|}
\hline
Alice      &         & Bob        \\ \hline
$\mvec{\sigma} \stackrel{r}{\leftarrow} \mathbb{F}_q^n$&         &            \\
$f \stackrel{r}{\leftarrow} \Lambda_{n,d}$&         &\\
$f(\mvec{\sigma}) = 0$ & $\xrightarrow{f}$& $\mvec{g} \stackrel{r}{\leftarrow} \Lambda_{n,1}^n$ \\
           &         & $\mvec{\psi} \stackrel{r}{\leftarrow} (\Lambda_{n,m}^n)^*$\\
           &         & $\mvec{r} \stackrel{r}{\leftarrow} \Lambda_{n,m-d}^n$\\
           &$\xleftarrow{\left(\mvec{g}, \mvec{c}\right)}$& $\mvec{c}:=\mvec{\psi}\circ\mvec{g}+f\mvec{r}$\\
$\mvec{s}:=\mvec{g}(\mvec{\sigma})$   &  & \\
$\mvec{u}:=\mvec{c}(\mvec{\sigma})$  &$\xrightarrow{\mvec{u}}$&  $\mvec{s}=\mvec{\psi}^{-1}(\mvec{u})$\\\hline
\end{tabular}
\caption{The original protocol}
\label{fig:previous_protocol}
\end{figure}

\subsection{The Modified Protocol in the Original Paper}
\label{subsec:previous_improvement}

In this section, we explain the improvement of the protocol above given in Section 4.2 of the original paper to increase the possibilities of multivariate polynomial map $\mvec{\psi}$.
In the original protocol, we restricted the polynomial map $\mvec{\psi}$ to be injective in order Bob to obtain the common key uniquely.
In contrast, here we use a general polynomial map $\mvec{\psi}$, and from the candidate set $\mvec{\psi}^{-1}(\mvec{u})$ of common keys, we exclude ones which do not satisfy the necessary condition $f = 0$.
Precisely, we change Step \ref{previous_algorithm__Bob} of the algorithm in the following manner:
\begin{itemize}
\item[(a)]
Compute the set $\mvec{\psi}^{-1}(\mvec{u})$.
\item[(b)]
If $\#\mvec{\psi}^{-1}(\mvec{u}) = 1$, then keep the $\mvec{s} \in \mvec{\psi}^{-1}(\mvec{u})$ as the common key and halt.
\item[(c)]
If $\#\mvec{\psi}^{-1}(\mvec{u}) \neq 1$, then compute all elements of $S := \{ \mvec{s} \in \mvec{\psi}^{-1}(\mvec{u})\mid f( \mvec{g}^{-1}(\mvec{s}) ) = 0\}$.
In other words, for each element $\mvec{s}$ of $\mvec{\psi}^{-1}(\mvec{u})$, check whether it satisfies that $f( \mvec{g}^{-1}(\mvec{s}) ) = 0$, and if not, exclude the $\mvec{s}$ from the set $S$.
\item[(d)]
If finally $\#S = 1$, then keep the element of $S$ as the common key; otherwise, restart from Step \ref{previous_algorithm__Alice}.
\end{itemize}

\subsection{Advantage of the Original Protocol}

The protocol constructs a ciphertext in a way different from the usual multivariate cryptography (using a central map and composing it with two affine maps), and consequently, there is a possibility to reduce the parameter size and the ciphertext size by avoiding known attacks to multivariate cryptosystems.
Also, the protocol was an improvement of Yosh's protocol \cite{Yosh} and succeeded in decreasing the degree of polynomials from exponential order to polynomial order, which improves the efficiency.

\subsection{Disadvantage of the Modified Protocol}
\label{sec:previous_protocol__disadvantage}

The improvement in Section \ref{subsec:previous_improvement} aimed at enhancing the security by enlarging the possibility of the map $\mvec{\psi}$.
However, even though an additional check using the condition $f(\mvec{\sigma})=0$ is introduced, there may be risk of failure of the protocol due to non-injectivity of $\mvec{\psi}$.
In fact, for the experiments of the protocol performed in the original paper \cite{Akiyama}, the highest average success rate with their proposed parameters was only $89.9\%$.
In contrast, practically desirable values of failure rates are of the order of $2^{-64}$ or even smaller.
Therefore, the success rate of the modified protocol has to be much improved.

\section{Our Proposed Protocol}
\label{sec:proposed_protocol}

\subsection{Protocol Description}

Similarly to \cite{Akiyama}, Alice and Bob are going to agree on a common key using a public channel.
In addition to the originally used parameters, we introduce parameters $p$ related to the range of the common key and $\ell$ determining the number of equations in $f$.
From now, we regard $\mathbb{Z}_p := \{0,1,\dots,p-1\} \subseteq \mathbb{F}_q$.
\begin{align*}
q \colon& \mbox{prime number which is the number of elements of the coefficient field}\\
p \colon& \mbox{integer related to the range of the common key}\\
n \colon& \mbox{number of variables}\\
m \colon& \mbox{degree of polynomials generated by Bob}\\
d \colon& \mbox{degree of polynomial generated by Alice}\\
\ell \colon& \mbox{number of polynomials generated by Alice}
\end{align*}
Our proposed protocol is as follows (see also Figure \ref{fig:proposed_protocol}):
\begin{enumerate}
\item \label{item:proposed_protocol_Alice}
Alice sends a system of multivariate polynomial equations $\mvec{f}(\mvec{x})=0$ to Bob, and keeps its solution $\underline{s}$ belonging to $\mathbb{Z}_p^n \subseteq \mathbb{F}_q^n$ as the common key.
The detail is as follows:
\begin{enumerate}
\item
Generate a uniformly random $\mvec{s} \in \mathbb{Z}_p^n$, which will be the common key.
\item
Generate a system of degree-$d$ polynomials $\tilde{f}(\mvec{x}) = (\tilde{f}_1(\mvec{x}),\dots,\tilde{f}_{\ell}(\mvec{x})) \in \mathbb{F}_q[\mvec{x}]^{\ell}$ uniformly at random.
\item
Compute $\mvec{f}(\mvec{x}) := (f_1(\mvec{x}),\dots,f_{\ell}(\mvec{x})) = \tilde{f}(\mvec{x})-\tilde{f}(\mvec{s})$.
\item
Send $\mvec{f}(\mvec{x})$ to Bob.
\end{enumerate}
\item
Bob sends a polynomial map $\mvec{c}(\mvec{x})$ to Alice.
The detail is as follows:
\begin{enumerate}
\item
Randomly generate a polynomial $\psi_j(\mvec{x}) \in \mathbb{F}_q[x_1,\dots,x_j]$ of degree $\deg \psi_j = m$ for each $1 \leq j \leq n$, and set $\mvec{\psi}(\mvec{x}) := (\psi_1(\mvec{x}),\dots,\psi_n(\mvec{x}))$ (see also Section \ref{subsec:proposed_protocol__polynomial_map}).
%
%
\item
Generate a polynomial map $\mvec{r}(\mvec{x}) = (r_1(\mvec{x}), \dots, r_n(\mvec{x})) \in \mathbb{F}_q[\mvec{x}]^n$ of degree $\deg r_j = m-d$ ($1 \leq j \leq n$) uniformly at random.
\item
Choose $(t_1,\dots,t_n) \in \{1,\dots,\ell\}^n$ uniformly at random.
\item
Compute $c_i(\mvec{x}) = \psi_i(\mvec{x}) + f_{t_i}(\mvec{x})r_i(\mvec{x})$ for each $1 \leq i \leq n$.
\item
Send $\mvec{c}(\mvec{x})$ to Alice.
\end{enumerate}
\item
Alice computes $\mvec{c}(\mvec{s}) = \mvec{u}$ and sends $\mvec{u}$ to Bob.
\item
Bob computes the common key $\mvec{s}$.
The detail is as follows:
\begin{enumerate}
\item
Compute the set $\mvec{\psi}^{-1}(\mvec{u}) \cap \mathbb{Z}_p^n$.
\item
If $\#\mvec{\psi}^{-1}(\mvec{u}) \cap \mathbb{Z}_p^n = 1$, then keep the $\mvec{s} \in \mvec{\psi}^{-1}(\mvec{u}) \cap \mathbb{Z}_p^n$ as the common key and halt.
\item
If $\#\mvec{\psi}^{-1}(\mvec{u}) \cap \mathbb{Z}_p^n \neq 1$, then compute all elements of $S = \{ \mvec{s} \in \mvec{\psi}^{-1}(\mvec{u}) \cap \mathbb{Z}_p^n \mid \mvec{f}(\mvec{s}) = 0 \}$.
\item
If $\#S = 1$, then keep the element of $S$ as the common key; otherwise, restart from Step \ref{item:proposed_protocol_Alice}.
\end{enumerate}
\end{enumerate}
\begin{figure}[t!]
  \centering
  \begin{tabular}{|lll|}
  \hline
  Alice      &         & Bob        \\ \hline
  $\mvec{s} \stackrel{r}{\leftarrow} \mathbb{Z}_p^n$&         &            \\
  $\mvec{f} \stackrel{r}{\leftarrow} \Lambda_{n,d}^{\ell}$&         &\\
  $\mvec{f}(\mvec{s}) = 0$ & $\xrightarrow{\mvec{f}}$&  $\psi_j \stackrel{r}{\leftarrow} \Lambda_{j,m}$, $\mvec{\psi} := (\psi_1,\dots,\psi_n)$\\
             &         & $\mvec{r} \stackrel{r}{\leftarrow} \Lambda_{n,m-d}^n$\\
             &         & $t_i \stackrel{r}{\leftarrow} \{1,\dots,\ell\}$\\
             &$\xleftarrow{\mvec{c}}$& $c_i:=\psi_i+f_{t_i}r_i$\\
  $\mvec{u}:=\mvec{c}(\mvec{s})$  &         &\\
  &$\xrightarrow{\mvec{u}}$&  $\mvec{s}\in \mvec{\psi}^{-1}(\mvec{u}) \cap \mathbb{Z}_p^n$\\
  &         & s.t.\ $\mvec{f}(\mvec{s})=0$\\\hline
\end{tabular}
\caption{Our proposed protocol}
\label{fig:proposed_protocol}
\end{figure}

\subsection{Construction of Easy-to-Invert Polynomials}
\label{subsec:proposed_protocol__polynomial_map}

In our proposed protocol, the total efficiency depends highly on the efficiency of generating the polynomial map $\mvec{\psi}$ and computing the inverse $\mvec{\psi}^{-1}$.
Therefore, it is important to use polynomial maps $\mvec{\psi}$ that can be efficiently generated and whose preimage can be efficiently computed.
Here we adopt polynomial systems where the number of variables in each polynomial is gradually incremented, such as in Gr\"{o}bner basis of a zero-dimensional ideal about lexicographical order (see an example below).
For such a polynomial system, its preimage can be recursively computed by solving univariate polynomial equations and then substituting the solutions to the remaining polynomials.
In our proposed protocol, we only need solutions belonging to $\mathbb{Z}_p^n$, therefore we have to only keep the solutions of the univariate equations belonging to $\mathbb{Z}_p$, and we can prune some branches when a solution in $\mathbb{Z}_p$ does not exist.
This reduces the computational cost drastically, in contrast to the original protocol where we needed to solve univariate equations $d^n$ times.

\begin{example}
Here we give a toy example of our polynomial systems in the case of degree two with three variables over coefficient field $\mathbb{F}_5$:
\begin{eqnarray*}
  \left\{
    \begin{array}{l}
  \psi_1(\underline{x}) = 3{x_1}^2 + x_1 +  4 \\
  \psi_2(\underline{x}) = {x_2}^2 + 2{x_1}{x_2} + 4 x_1 + x_2 + 3  \\
  \psi_3(\underline{x}) = 4{x_1}^2 + 2{x_3}^2 + {x_1}{x_3} + 3{x_2}{x_3} + 1
\end{array}
\right.
\end{eqnarray*}
\end{example}

\subsection{Theoretical Estimate of Failure Probability}
\label{subsec:error_bound_analysis}

Here, for parameters $(q,p,n,m,d,\ell)$, we estimate the failure probability of our proposed protocol.
Here the \lq\lq failure\rq\rq{} means the case where two or more candidates remain after the computation by Bob to determine the common key (i.e., the protocol is restarted at the final step).
In order to estimate the failure probability, we analyze the expected number of candidates for the common key computed by Bob.
Recall that Bob's computation at the step consists of the following two steps:
\begin{enumerate}
\item \label{error_analysis_step1}
Computing the preimage of $\mvec{\psi}$ in $\mathbb{Z}_p^n$.
\item \label{error_analysis_step2}
From the candidates obtained at Step \ref{error_analysis_step1}, excluding ones that do not satisfy the condition $\mvec{f} = 0$.
\end{enumerate}
We divide the argument into the two steps above.
From now, we focus on the case $m = 2$ and $d = 1$ which are our proposed parameters.

\paragraph{For Step \ref{error_analysis_step1}.}

Recall that now $\psi_1$ is a quadratic polynomial in $x_1$, $\psi_2$ is a quadratic polynomial in $x_1,x_2$, and so on, and $\psi_n$ is a quadratic polynomial in $x_1,\dots,x_n$.
Consequently, we can solve the system of equations by recursively solving univariate quadratic equations.
We represent this process by using a rooted tree structure from $0$-th level (root) to $n$-th level, where a node at $k$-th level corresponds to a partial solution $(s_1,\dots,s_k)$ obtained from the first $k$ polynomials $\psi_1,\dots,\psi_k$.
Hence the nodes at $n$-th level represent the candidates for the common key computed by Bob.
Now for each node $(s_1,\dots,s_k)$ at $k$-th level ($0 \leq k \leq n - 1$), there are the following three possibilities;
\begin{itemize}
\item
the node has no child nodes (that is, the quadratic equation $\psi_{k+1}(s_1,\dots,s_k,x_{k+1})$ has no solution $x_{k+1} \in \mathbb{Z}_p$);
\item
the node has only one child node (that is, the quadratic equation $\psi_{k+1}(s_1,\dots,s_k,x_{k+1})$ has a unique solution $x_{k+1} \in \mathbb{Z}_p$);
\item
the node has two child nodes (that is, the quadratic equation $\psi_{k+1}(s_1,\dots,s_k,x_{k+1})$ has two different solutions $x_{k+1} \in \mathbb{Z}_p$).
\end{itemize}

Here we note that there is always at least one path from the root to a node at $n$-th level, which corresponds to the \lq\lq correct\rq\rq{} solution chosen by Alice.
We call it the correct path.
From now, we evaluate (by using some heuristic assumptions) an upper bound for the expected number of \lq\lq incorrect\rq\rq{} solutions, that is, paths from the root to a node at $n$-th level different from the correct path.

First we consider the case that the correct path has another branch at $k$-th level ($0 \leq k \leq n - 1$).
We note that the current equation $\psi_{k+1}(s_1,\dots,s_k,x_{k+1})$ has at least one solution in $\mathbb{F}_q$ (which is the correct solution), therefore it has two solutions in $\mathbb{F}_q$.
By heuristically assuming that the other solution is uniformly random over $\mathbb{F}_q$, the probability of branching, i.e., the probability that the other solution is in $\mathbb{Z}_p$ and is different from the correct solution, is at most $p/q$.

Secondly, in a situation where a node, say $v$, at $k$-th level ($1 \leq k \leq n - 1$) that is not on the correct path exists, we evaluate the expected number of child nodes of $v$.
To simplify the argument, here we heuristically assume that the behavior of child nodes of $v$ is independent of the behaviors at the previous levels.
Now it seems not easy for evaluating the probability that the current univariate quadratic equation has a solution in $\mathbb{F}_q$ (note that if it has a solution in $\mathbb{F}_q$, then it has two solutions in $\mathbb{F}_q$ possibly with multiplicity); to derive an upper bound, here we just bound the probability from above by $1$.
We also heuristically assume that now the two solutions distribute independently and uniformly at random over $\mathbb{F}_q$.
Under the assumption, the probability that the node $v$ has a first child node (that is, at least one of the two solutions belongs to $\mathbb{Z}_p$) is given by $1 - (1 - p/q)^2 = 2p / q - (p/q)^2 \leq 2p / q$; and the probability that the node $v$ has the second child node (that is, both of the two solutions belong to $\mathbb{Z}_p$ and these are different) is given by $(p/q) \cdot (p-1)/q \leq (p/q)^2$.
Hence, the expected number of child nodes of $v$ is upper bounded by $\alpha := 2p / q + (p/q)^2$.

By heuristically assuming that the behavior of each level is independent of each other, the expected number of nodes at $n$-th level appearing after branching from the correct path at $k$-th level (now the node $v$ above is at $(k+1)$-th level) is upper bounded by $\alpha^{n-(k+1)}$.
Therefore, the expected number of incorrect solutions is upper bounded by
\[
\sum_{k=0}^{n-1} \frac{ p }{ q } \cdot \alpha^{n-(k+1)}
= \frac{ p }{ q } \cdot \frac{ 1 - \alpha^n }{ 1 - \alpha}
= \frac{ p }{ q } \cdot \frac{ 1 - (2p / q + (p/q)^2)^n }{ 1 - (2p / q + (p/q)^2) } \enspace.
\]

\paragraph{For Step \ref{error_analysis_step2}.}

To simplify the argument, we heuristically assume that the values of $f_i(\mvec{s}')$ ($1 \leq i \leq \ell$) for an incorrect solution $\mvec{s}'$ are uniformly random over $\mathbb{F}_q$ and independent of each other.
Under the assumption, the probability that an incorrect solution $\mvec{s}'$ satisfies that $\mvec{f}(\mvec{s}') = 0$ is $1/q^{\ell}$.

\paragraph{}
Summarizing, by writing the number of incorrect candidates for the common key computed by Bob as $X$, we have
\[
\mathbb{E}[X]
\leq \frac{ p }{ q } \cdot \frac{ 1 - (2p / q + (p/q)^2)^n }{ 1 - (2p / q + (p/q)^2) } \cdot \frac{ 1 }{ q^{\ell} } \enspace.
\]
Our proposed protocol fails if and only if $X \geq 1$, and by Markov's Inequality, its probability is bounded by
\begin{equation}
\label{eq:theoretical_upper_bound}
Pr[X \geq 1]
\leq \mathbb{E}[X]
\leq \frac{ p }{ q } \cdot \frac{ 1 - ( 2p/q + (p/q)^2 )^n }{ 1 - ( 2p/q + (p/q)^2 ) } \cdot \frac{ 1 }{ q^{\ell} } \enspace.
\end{equation}
By substituting our choice of parameters
\[
(q,p,n,m,d,\ell) = (46116646144580573897,19,32,2,1,1)
\]
into the formula above, we obtain an estimated upper bound $8.93 \times 10^{-39} \approx 1.52 \times 2^{-127}$ for the failure probability.

\section{Experimental Results}
\label{sec:experimental_results}

In this section, we explain our experimental results on our proposed protocol.
We used a PC with 8 GB memory and 2 GHz Intel Core i5, and used Magma for implementation.

\subsection{Confirmation of the Theoretical Upper Bound}

In order to confirm that our theoretical upper bound in Eq.\eqref{eq:theoretical_upper_bound} under several heuristic assumptions is not too optimistic, we executed our proposed protocol many times and observed how many times a failure occurs.
Here we used much smaller parameters than our proposed parameter set, as the original parameter yields too small estimated failure probability and therefore it is not feasible to confirm it experimentally.
In detail, the choices of $(p,n,m,d,\ell) = (19,32,2,1,1)$ are the same as our proposed parameter set, while we change the parameter $q$ to make the estimated failure probability fairly high.
We used three choices of $q$ as in Table \ref{tab:experiment_error_probability}, and performed $1000$ protocol executions for each choice of $q$.
The resulting error ratios as well as the estimated error probabilities by using Eq.\eqref{eq:theoretical_upper_bound} are shown in the table.
At least this experimental result does not contradict the theoretical estimate.

\begin{table}[t]
\centering
\caption{Error ratios with various choices of (small) $q$, in $1000$ trials for each $q$}
\label{tab:experiment_error_probability}
\begin{tabular}{|c|c|c|c|} \hline
$(p,n,m,d,\ell)$ & $q$ & error ratio & estimated error prob. \\ \hline
& $53$ & $1.20 \times 10^{-2}$ & $4.36 \times 10^{-2}$ \\ 
$(19,32,2,1,1)$ & $71$ & $8.00 \times 10^{-3}$ & $9.59 \times 10^{-3}$ \\ 
& $97$ & $1.00 \times 10^{-3}$ & $3.54 \times 10^{-3}$ \\ \hline
\end{tabular}
\end{table}

\subsection{Experiments with Parameters Similar to the Modified Protocol}
\label{sec:experimental_results__similar_parameter}

For the sake of comparison, as a parameter set similar to $(q,n,m,d)=(4,25,2,1)$ used in the modified protocol in \cite{Akiyama}, we performed experiments using parameter $(q,p,n,m,d,\ell)=(7,2,32,2,1,1)$.
We executed our protocol $300$ times, and the protocol succeeded $287$ times, therefore the success ratio was about $95.6\%$.
This improves the success ratio $89.9\%$ of the protocol mentioned in Section \ref{sec:previous_protocol__disadvantage}.
We note that the theoretical upper bound in Eq.\eqref{eq:theoretical_upper_bound} of the failure probability with this parameter becomes $0.118$, which indeed bounds the experimental failure ratio $0.044$.

We also performed experiments using another parameter $(q,p,n,m,d,\ell) = (7,4,32,2,1,1)$; in this case, the protocol succeeded only $20$ times among $300$ trials.
The significantly lower success ratio would be caused by the property that now the ratio $p / q$ of the range of the correct solution among the whole coefficient field becomes too large.
In fact, now the theoretical upper bound in Eq.\eqref{eq:theoretical_upper_bound} of the failure probability becomes $3.88 \times 10^5$ which is a meaningless value.

\subsection{Computational Time for Our Protocol}

Table \ref{tab:execution_time} shows the computational times of our protocol with parameter
\[
(q,n,m,d,\ell) = (46116646144580573897,32,2,1,1)
\]
and various choices of $p = 19$, $19^2$, $19^3$.
The execution time increased when $p$ becomes larger, but the change of execution times among these choices of $p$ is not significantly large.
The reason would be that now the ratio $p/q$ is too small (e.g., $p/q = 1.49 \times 10^{-16}$ when $p = 19^3$) to affect the number of nodes in the tree (that is, the total number of equations to be solved).
\begin{table}[t!]
  \centering
  \caption{Computational times for our proposed protocol with $p = 19$, $19^2$, and $19^3$}
  \label{tab:execution_time}
  \begin{tabular}{|c|c|c|} \hline
      $p$& total time for $1000$ executions (s) & average time (s) \\ \hline 
      $19$ & $71.460$ & $7.15 \times 10^{-2}$\\
      $19^2$ & $72.800$ & $7.28 \times 10^{-2}$\\
      $19^3$ & $239.370$ & $2.39 \times 10^{-1}$\\ \hline
  \end{tabular}
\end{table}

\subsection{Relation between Failure Ratios and Parameter $\ell$}

One of the main difference of our proposed protocol from the previous protocol is that now we may use $\ell \geq 1$ conditions $\mvec{f} = 0$, not only a single condition, to exclude the incorrect common keys.
Table \ref{tab:probability_with_various_l} shows the numbers of success (among $1000$ trials) in our experiments with two parameter sets $(q,p,n,m,d) = (5,2,32,2,1)$ and $(3,2,32,2,1)$ and various choices of $\ell = 1,\dots,5$.
The result shows that the failure ratio decreases when $\ell$ increases, which is consistent with the theoretical estimate given in Section \ref{subsec:error_bound_analysis}.
\begin{table}[t!]
\centering
\caption{Numbers of success with various choices of $\ell$ ($1000$ trials for each parameter)}
\label{tab:probability_with_various_l}
  \begin{tabular}{|c|c|c|} \hline
    $(q,p,n,m,d)$ & $(5,2,32,2,1)$ & $(3,2,32,2,1)$\\ \hline 
    $\ell = 1$ & $806$ &  $26$\\
    $\ell = 2$ & $950$ & $186$\\
    $\ell = 3$ & $993$ & $552$\\ 
    $\ell = 4$ & $996$ & $802$\\
    $\ell = 5$ & $999$ & $945$\\ \hline
  \end{tabular}
\end{table}

\subsection{Degree of Regularity}
\label{sec:experimental_result__degree_of_regularity}

Here we explain our experiments about the degree of regularity used in Section \ref{sec:security_evaluation__Grobner_basis}.
If we consider to solve the system of equations $\mvec{f}(\mvec{x}) = 0$ and $\mvec{c}(\mvec{x}) = \mvec{u}$ in $\mathbb{F}_q^n$ instead of $\mathbb{Z}_p^n$, the corresponding ideal is generated by $\mvec{f}(\mvec{x})$, $c_1(\mvec{x}) - u_1, \dots, c_n(\mvec{x}) - u_n$.
In order to add the constraint that the solution has to be found in $\mathbb{Z}_p^n$, we introduce the following polynomial system $\mvec{h}$:
\begin{eqnarray*}
h_i(\mvec{x}) = \prod_{\gamma = 0}^{p-1} (x_i - \gamma) \mbox{ ($1 \leq i \leq n$)} \enspace.
\end{eqnarray*}
Now the set of solutions in $\mathbb{Z}_p^n$ corresponds to the ideal $\mathcal{I}$ generated by $\mvec{f}(\mvec{x})$, $c_1(\mvec{x}) - u_1, \dots, c_n(\mvec{x}) - u_n$, and $\mvec{h}(\mvec{x})$.
We fix parameters $(q,m,d,\ell) = (46116646144580573897,2,1,1)$, and calculated the degree of regularity for the cases $2 \leq n \leq 10$ and $p \in \{2,3,4,19\}$.
The result was that the degree of regularity of the ideal $\mathcal{I}$ is always $d_{reg} = n + 1$.
Due to the result, we can expect that the degree of regularity of any ideal of this type would be $d_{reg} = n + 1$ regardless of whether $p$ is larger than $n$ or smaller than $n$; this observation will be applied in Section \ref{sec:security_evaluation__Grobner_basis} to the case $p = 19$ and $n = 32$.

\section{Comparison with Previous Work}
\label{sec:comparison}

\subsection{Failure Probabilities of Protocols}

As shown in Section \ref{subsec:error_bound_analysis}, the failure probability of our proposed protocol with the proposed parameter set is of the order of $10^{-39}$, which is significantly lower than the experimentally derived failure ratio $10.1\%$ of the modified protocol in \cite{Akiyama}.
We emphasize that this improvement is not just due to different choices of parameters; the argument in Section \ref{sec:experimental_results__similar_parameter} showed that our proposed protocol also improves the modified protocol in \cite{Akiyama} even with the choice of similar parameters.

For a theoretical estimate of the error probability in the modified protocol in \cite{Akiyama}, we roughly assume that for any $\mvec{s} \in \mvec{\psi}^{-1}(\mvec{u})$ that is not the correct common key, the vector $g^{-1}(\mvec{s})$ distributes uniformly at random, and hence $f(g^{-1}(\mvec{s}))$ is also uniformly random.
Under the assumption, by Bob's step to check whether $f(g^{-1}(\mvec{s})) = 0$ or not, the $\mvec{s}$ accidentally passes the check with probability $1/q$.
Therefore, assuming further that the size of $\mvec{\psi}^{-1}(\mvec{u})$ is bounded by a small constant, the error probability of the modified protocol in \cite{Akiyama} will be $\approx 1/q$.
This estimate seems to be consistent to the experimental result in Table 4 of \cite{Akiyama}; see also the description before Table 4 of \cite{Akiyama}.
In contrast, for our theoretical estimate of the error probability in Eq.\eqref{eq:theoretical_upper_bound} with $\ell = 1$, when $q$ is much larger than $p$ and $n$ is fairly large, ignoring the term $(2p/q + (p/q)^2)^n$ gives an approximation
\[
\frac{ p }{ q } \cdot \frac{ 1 }{ 1 - (2p/q + (p/q)^2) } \cdot \frac{ 1 }{ q }
= \frac{ p }{ q^2 - (2pq + p^2) }
\approx \frac{ 1 }{ q } \cdot \frac{ p }{ q - 2p } \enspace.
\]
This suggests that the error probability of our proposed protocol will be significantly lower than the modified protocol in \cite{Akiyama} if the order of the parameter $q$ is similar and $p$ is much smaller than $q$ (as in our proposed parameter).

\subsection{Necessary Numbers of Communication Rounds}
\label{sec:comparison__rounds}

Here we suppose that the \lq\lq practical failure probability\rq\rq{} of key exchange protocols is $2^{-120} \approx 7.52 \times 10^{-37}$, taken from the failure probability with a parameter set LightSaber-KEM of SABER \cite{SABER} which is one of the Round 3 Finalists of NIST PQC standardization.
The failure probability of one execution of the modified protocol in \cite{Akiyama} is $0.101$, therefore $37$ trials are needed to achieve the overall failure probability $2^{-120}$.
One protocol execution uses two communication rounds between Alice and Bob, therefore the total number of communication rounds is $2 \times 37 = 74$.
In contrast, our proposed protocol with the proposed parameter set already achieves failure probability $1.52 \times 2^{-127} < 2^{-120}$, therefore the required number of communication rounds is just two.

\subsection{Communication Complexity}

We compare the amount of communication bits between Alice and Bob, for our proposed protocol and the modified protocol in \cite{Akiyama}.
Here we express a polynomial as a vector of its coefficients; a polynomial of degree $d$ with $n$ variables is represented by a vector of dimension $\sum_{k=0}^{d} \binom{k+n-1}{k}$.
Also, we suppose that an element of $\mathbb{F}_q$ is represented by $\log_{2}q$ bits.

In the modified protocol in \cite{Akiyama}, the communicated objects are $f = \mvec{f}$, $\mvec{g}$, $\mvec{c}$, and $\mvec{u}$.
Here $f$ is a single degree-$d$ polynomial with $n$ variables over $\mathbb{F}_q$; $\mvec{g}$ consists of $n$ degree-$1$ polynomials with $n$ variables (each having $n + 1$ coefficients) over $\mathbb{F}_q$; $\mvec{c}$ consists of $n$ degree-$m$ polynomials with $n$ variables over $\mathbb{F}_q$; and $\mvec{u}$ is an $n$-dimensional vector over $\mathbb{F}_q$.
Therefore, the numbers of communicated elements of $\mathbb{F}_q$ are as in the left column of Table \ref{tab:communication_complexity}.
\begin{table}[t!]
\centering
\caption{Comparison of communication complexity}
\label{tab:communication_complexity}
  \begin{tabular}{|c|c|c|} \hline
     & protocol in \cite{Akiyama} ($\times \log_2 q$ bits) & our protocol ($\times \log_2 q$ bits) \\ \hline 
    $\mvec{f}$ & $\sum_{k=0}^{d} \binom{k+n-1}{k}$ & $\ell \cdot \sum_{k=0}^{d} \binom{k+n-1}{k}$ \\
    $\mvec{g}$ & $n^2+n$ & --- \\
    $\mvec{c}$ & $n \cdot \sum_{k=0}^{m} \binom{k+n-1}{k}$ & $n \cdot \sum_{k=0}^{m} \binom{k+n-1}{k}$ \\ 
    $\mvec{u}$ & $n$ & $n$\\
    total & $\sum_{k=0}^{d} \binom{k+n-1}{k} + n \cdot \sum_{k=0}^{m} \binom{k+n-1}{k} + n^2 + 2n$ & $\ell \cdot \sum_{k=0}^{d} \binom{k+n-1}{k} + n \cdot \sum_{k=0}^{m} \binom{k+n-1}{k} + n$ \\ \hline
  \end{tabular}
\end{table}

In our proposed protocol, the changes from the modified protocol in \cite{Akiyama} are the following two points; the number of polynomials in $\mvec{f}$ becomes $\ell$ instead of one; and $\mvec{g}$ is not communicated.
Therefore, the numbers of communicated elements of $\mathbb{F}_q$ are as in the right column of Table \ref{tab:communication_complexity}.

For a parameter set $(q,n,m,d) = (9,50,2,1)$ used in the original paper \cite{Akiyama}, the communication complexity of one execution of the modified protocol in \cite{Akiyama} is $2.19 \times 10^5$ bits.
As the protocol in \cite{Akiyama} needs $37$ trials to achieve the practical failure probability (see Section \ref{sec:comparison__rounds}), the total communication complexity becomes $8.10 \times 10^6$ bits.
In contrast, our proposed protocol with the proposed parameter
\[
(q,p,n,m,d,\ell) = (46116646144580573897,19,32,2,1,1)
\]
achieves the practical failure probability by only one execution, and the communication complexity is $1.18 \times 10^6$ bits.
Hence our proposed protocol improves the communication complexity compared to the modified protocol in \cite{Akiyama}.

\section{Security Evaluation}
\label{sec:security_evaluation}

In this section, we analyze the security of our proposed protocol.

\subsection{Constrained PME Problem}
\label{subsec:constrained_PME_problem}

Here we define a computational problem named constrained PME problem, which is a modification of PME problem described in Section \ref{subsec:PME_problem} and is a base of the security of our proposed protocol.

\begin{problem}
[constrained PME problem]
For $(f_1(\mvec{x}), \dots, f_{\ell}(\mvec{x}), c_1(\mvec{x}), \dots, c_n(\mvec{x})) \in \mathbb{F}_q[\mvec{x}]^{n + \ell}$ and $(u_1, \dots, u_n) \in \mathbb{F}_q^n$, constrained PME problem is a problem of finding a solution to the system of multivariate polynomial equations that belongs to $\mathbb{Z}_p^n$:
\begin{eqnarray*}
  \begin{cases}
    f_1(x_1, \dots, x_n) = 0 \\
    \qquad\vdots\\
    f_{\ell}(x_1, \dots, x_n) = 0 \\
    c_1(x_1, \dots, x_n) = u_1\\
    \qquad\vdots\\
    c_n(x_1, \dots, x_n) =u_n \enspace.
  \end{cases}
\end{eqnarray*}
\end{problem}

\subsection{Security against Key Recovery Attacks}

As well as the original paper \cite{Akiyama}, in this paper, we consider the security against the Key Recovery Attack by Honest Passive Observer (KRA-HPO).
In the scenario, an attacker is supposed to not interfere with the communication between Alice and Bob.
Then we prove that our proposed protocol is secure in this sense by assuming the hardness of constrained PME problem defined above.

To formulate the security, let $\Sigma$ denote our proposed key exchange protocol, as described in Figure \ref{fig:our_protocol_for_security_definition}.
Let $\mathcal{A}$ denote an adversary's attack algorithm.
Then we define the security experiment for KRA-HPO adversary as in Figure \ref{fig:KRA-HPO_experiment}, where $\mathsf{Gen}$ denotes an algorithm to generate the parameters of our proposed protocol.
Now we define the security as follows:
\begin{figure}[t!]
\centering
\begin{tabular}{|l|}
\hline
$\Sigma(q,p,n,m,d,\ell)$              \\ \hline
$\mvec{f}$ $\stackrel{r}{\leftarrow}$ $\mathbb{F}_q[\mvec{x}]^{\ell};$\\
$f(\mvec{s})$ $=$ $0, \mvec{s} \stackrel{r}{\leftarrow} \mathbb{Z}_p^n;$\\
$\mvec{r}$ $\stackrel{r}{\leftarrow}$ $\Lambda_{n, m-d}^{n};$\\
$\mvec{\psi}$ $\stackrel{r}{\leftarrow}$ $\Lambda_{n,m}^n;$\\
$t_i$ $\stackrel{r}{\leftarrow}$ $\{1,\dots,\ell\};$\\
$c_i$ $\leftarrow$ $\psi_i + f_{t_i}\cdot r_i;$\\
$\mvec{u}$ $\leftarrow$ $\mvec{c}(\mvec{s});$\\
Output $(\mvec{s},\mvec{f},\mvec{c},\mvec{u})$\\
\hline
\end{tabular}
\caption{Our proposed key exchange protocol $\Sigma$}
\label{fig:our_protocol_for_security_definition}
\end{figure}
\begin{figure}[t!]
\centering
\begin{tabular}{|l|}
\hline
$\mathsf{Exp}_{\Sigma, \mathcal{A}}^{\mathrm{KRA-HPO}}(\kappa)$        \\ \hline
$(q,p,n,m,d,\ell)\stackrel{r}{\leftarrow} \mathsf{Gen}(1^{\kappa});$\\
$(\mvec{s},\mvec{f},\mvec{c},\mvec{u})\stackrel{r}{\leftarrow}\Sigma(q,p,n,m,d,\ell);$\\
$\mvec{s}{}' \leftarrow\mathcal{A}(\mvec{f},\mvec{c},\mvec{u});$\\
Output $( \mvec{s}, \mvec{s}{}' )$\\
\hline
\end{tabular}
\caption{Security experiment for KRA-HPO adversary $\mathcal{A}$ against protocol $\Sigma$}
\label{fig:KRA-HPO_experiment}
\end{figure}

\begin{definition}
We say that the key exchange protocol $\Sigma$ is KRA-HPO secure when for any probabilistic polynomial-time (PPT) adversary $\mathcal{A}$, its advantage defined below is negligible in the security parameter $\kappa$:
\begin{eqnarray*}
\mathsf{Adv}_{\Sigma, \mathcal{A}}^{\mathrm{KRA-HPO}}(\kappa) := \Pr[\mvec{s} = \mvec{s}{}' \mid (\mvec{s}, \mvec{s}{}' ) \leftarrow \mathsf{Exp}_{\Sigma, \mathcal{A}}^{\mathrm{KRA-HPO}}(\kappa)] \enspace.
\end{eqnarray*}
\end{definition}

We also put the following assumption on the hardness of constrained PME problem:

\begin{assumption}
[hardness of constrained PME problem]
Suppose that $\mvec{f}$, $\mvec{c}$, and $\mvec{u}$ are chosen as in our proposed protocol.
We assume that for any PPT algorithm $\mathcal{B} = \mathcal{B}(1^{\kappa},\mvec{f},\mvec{c},\mvec{u})$, the probability that its output $\mvec{s}{}'$ satisfies that $\mvec{f}(\mvec{s}{}') = 0$ and $\mvec{c}(\mvec{s}{}') = \mvec{u}$ is negligible in $\kappa$.
\end{assumption}

Then we have the following theorem:

\begin{theorem}
Under the assumption on the hardness of constrained PME problem explained above, our proposed protocol $\Sigma$ is KRA-HPO secure.
In more detail, if a PPT adversary $\mathcal{A}$ breaks the security of $\Sigma$ with advantage $\mathsf{Adv}_{\Sigma, \mathcal{A}}^{\mathrm{KRA-HPO}}(\kappa)$, then there exists a PPT algorithm $\mathcal{B}$ solving the constrained PME problem with probability at least $\mathsf{Adv}_{\Sigma, \mathcal{A}}^{\mathrm{KRA-HPO}}(\kappa)$.
\end{theorem}
\begin{proof}
Given an adversary $\mathcal{A}$ as in the statement, we define an algorithm $\mathcal{B}(1^{\kappa},\mvec{f},\mvec{c},\mvec{u})$ as follows: it runs $\mathcal{A}(\mvec{f},\mvec{c},\mvec{u})$ to obtain $\mvec{s}{}'$, and outputs the $\mvec{s}{}'$.
This $\mathcal{B}$ is PPT as well as $\mathcal{A}$.
Now the condition $\mvec{s}{}' = \mvec{s}$ holds in the experiment $\mathsf{Exp}_{\Sigma, \mathcal{A}}^{\mathrm{KRA-HPO}}(\kappa)$ with probability $\mathsf{Adv}_{\Sigma, \mathcal{A}}^{\mathrm{KRA-HPO}}(\kappa)$, and if it happens, then the output of $\mathcal{B}$ also satisfies that $\mvec{s}{}' = \mvec{s}$ and (as $\mvec{f}(\mvec{s}) = 0$ and $\mvec{c}(\mvec{s}) = \mvec{u}$ by the definition) hence $\mvec{f}(\mvec{s}{}') = 0$ and $\mvec{c}(\mvec{s}{}') = \mvec{u}$, i.e., $\mathcal{B}$ solves the constrained PME problem.
Hence the claim holds.
\end{proof}

\subsection{On Solving the Problem Using Gr\"{o}bner Basis}
\label{sec:security_evaluation__Grobner_basis}

In order to confirm our security assumption described above, here we consider to solve the constrained PME problem by computing Gr\"{o}bner basis of the ideal $\mathcal{I}$ generated by $\mvec{f}(\mvec{x})$, $c_1(\mvec{x}) - u_1,\dots,c_n(\mvec{x}) - u_n$, and $\mvec{h}(\mvec{x})$, where the polynomials $\mvec{h}$ are as defined in Section \ref{sec:experimental_result__degree_of_regularity}.
Here we focus on the case $\ell = 1$; choosing larger $\ell$ will decrease the hardness of the problem (see also Section \ref{sec:security_evaluation__exhaustive_search}).
The experimental result described in Section \ref{sec:experimental_result__degree_of_regularity} suggests that the system of polynomial equations $f_1(\mvec{x}) = 0$, $c_1(\mvec{x}) - u_1 = 0,\dots, c_n(\mvec{x}) - u_n = 0$, $\mvec{h}(\mvec{x}) = 0$ would be semi-regular; if the conjecture is true, then we have the degree of regularity $d_{reg} = 33$ for $\mathcal{I}$.
Now when we want to compute a Gr\"{o}bner basis of $\mathcal{I}$ by using $F_5$ algorithm, the computational complexity is estimated (as in Theorem \ref{thm:complexity_of_F5}) as the order of $\binom{n + d_{reg}}{n}^{\omega}$.
By substituting $n = 32$, $d_{reg}= 33$, and by using an estimate $\omega = 2.3$ as in \cite{omega}, the value becomes
\[
\binom{n + d_{reg}}{n}^{\omega}
\approx 4.8 \times 10^{42} \enspace.
\]
This value is significantly larger than $2^{128} \approx 3.4 \times 10^{38}$.
Hence it is expected that our proposed protocol would be secure in the sense of $128$-bit security against this kind of attacks.

\subsection{On Attacks Using Linear Algebra}

Here we consider a kind of attacks to recover the polynomial map $\mvec{\psi}$; if it were possible, then the adversary could compute the set $\mvec{\psi}^{-1}(\mvec{u}) \cap \mathbb{Z}_p^n$ which contains the common key $\mvec{s}$.

For the choice of parameter $\ell = 1$ mainly used in this paper, we note that by the construction of the protocol, the following relation holds:
\begin{eqnarray*}
\psi_i(\underline{x}) + f(\underline{x})r_i(\underline{x}) = c_i(\underline{x}) \mbox{ for } i=1, \dots, n \enspace.
\end{eqnarray*}
Here, for each $i$, the coefficients of $f$ and $c_i$ are known and the coefficients of $\psi_i$ and $r_i$ are not known.
This situation can be regarded as a system of linear equations.
Now the number of equations and the number of unknown coefficients in $\psi_i$ are equal, therefore the dimension of the solution space is equal to the number, say $N$, of unknown coefficients in $r_i$.
As $\deg r_i = m-d$, we have
\[
N = \sum_{k=0}^{m-d} \binom{k+n-1}{k} \enspace.
\]
With our proposed parameters
\[
(q,p,n,m,d,\ell) = (46116646144580573897,19,32,2,1,1) \enspace,
\]
we have $N = n + 1 = 33$ and $q^N = 8.08 \times 10^{648}$.
Therefore the number of candidates for the solution is significantly larger than $2^{128}$ and hence our proposed protocol is secure in the sense of $128$-bit security against this kind of attacks.

\subsection{On the Exhaustive Search with the Help of $\mvec{f}$}
\label{sec:security_evaluation__exhaustive_search}

Considering the exhaustive search for the common key $\mvec{s}$ over the range $\mathbb{Z}_p^n$, when we use parameters $\ell = 1$ and $d = 1$ as above, the number of candidates for $\mvec{s}$ can be reduced by a factor of $1/p$.
Indeed, when the coefficient of $x_i$ in $f$ is non-zero, for each value of $s_1,\dots,s_{i-1},s_{i+1},\dots,s_n$, the condition $f(\mvec{s}) = 0$ yields at most one possible value of $s_i \in \mathbb{Z}_p$.
Due to the observation, when we want to keep $128$-bit security, we have to compare $2^{128}$ with $p^{n-1}$ instead of $p^n$.
With our proposed parameters as above, we have $p^{n-1} = 19^{31} = 4.38 \times 10^{39} > 2^{128} = 3.40 \times 10^{38}$.
Therefore our proposed protocol is secure in the sense of $128$-bit security against this kind of attacks.

When the parameter $\ell$ is increased, there is an advantage that the upper bound of the failure probability decreases as shown in Eq.\eqref{eq:theoretical_upper_bound}.
However, there is also a disadvantage that the larger number of conditions in $\mvec{f}$ restricts the range of the common key further, which may make the exhaustive search easier.

\subsection*{Acknowledgements.}

The authors thank Tsuyoshi Takagi, Momonari Kudo, Hiroki Furue, and Yacheng Wang for their precious advice for this research.
This research was supported by the Ministry of Internal Affairs and Communications SCOPE Grant Number 182103105.

\end{document}